\documentclass[prl,twocolumn,showpacs,superscriptaddress,preprintnumbers,floatfix]{revtex4}

\usepackage{dcolumn}
\usepackage{amsmath}
\usepackage{amssymb}

\usepackage[dvips]{graphicx}
\usepackage{bm}   
\usepackage{verbatim}
\usepackage{amsthm}
\usepackage{color}

\newtheorem{thm}{Theorem}

\def\future {X_0^{N-1}}
\def\past {X_{-N}^{-1}}

\def\States {{\cal S}} 
\def\State {{ S}}    
\def\state {s}
\def\lastS {{S}_{i-1}}
\def\nextS {{S}_i}

\def\StateN {{S}_N}
\def\Sone {{S}_{-1}}
\def\Cmu   {C_\mu}

\def\EE   {{\bf E} }

\def\h   {{h} }
\def\hR   {{h^R} }
\def\ier { h_{erase}}
\def\LL { H_{erase}}
\def\eM {{$\epsilon$-machine}}
\def\Abet   {\cal{A} }

\def\PI {predictive information}

\begin{document}

\title{Information-theoretic bound on the energy cost of stochastic simulation}

\author{Karoline Wiesner}
\email{k.wiesner@bristol.ac.uk}
\affiliation{School of Mathematics, Centre for Complexity Sciences, University of Bristol, University Walk,  Bristol BS8 1TW,
United Kingdom}
\author{Mile Gu}
\author{Elisabeth Rieper}
\affiliation{Centre for Quantum Technologies, National University of
Singapore, 3 Science Drive 2, S15-03-18, Singapore 117543, Singapore}
\author{Vlatko Vedral}
\affiliation{Atomic and Laser Physics, Clarendon Laboratory, University of
Oxford, Parks Road, Oxford OX1 3PU, United Kingdom}
\affiliation{Department of Physics, National University of Singapore, 2
Science Drive 3, Singapore 117543, Singapore}
\affiliation{Centre for Quantum Technologies, National University of
Singapore, 3 Science Drive 2, S15-03-18, Singapore 117543, Singapore}

\date{\today}

\bibliographystyle{unsrt}

\begin{abstract}

Physical systems are often simulated using a stochastic computation where different final states result from identical initial states. Here, we derive the minimum energy cost of simulating a complex data set of a general physical system with a stochastic computation. We show that the cost is proportional to the difference between two information-theoretic measures of complexity of the data -- the \emph{statistical complexity} and the \emph{{\PI}}. 
We derive the difference as the amount of information erased during the computation. Finally, we illustrate the physics of information by implementing the stochastic computation as a Gedankenexperiment of a Szilard-type engine. The results create a new link between thermodynamics, information theory, and complexity.

 \begin{center}

arxiv.org/phys/0905.2918
 \end{center}
\end{abstract}

\maketitle

The idea of physics as information has a long history. The concept of entropy, at the heart of information theory, originated in the theory of thermodynamics. It was  Maxwell and Boltzmann who, in the beginning of the 19$^{th}$ century, recognized the intricate link between probability
distributions over configurations and thermodynamics. This laid the foundation to the field of statistical mechanics. The similarity between the thermodynamic entropy and the information entropy, introduced in 1948 by Shannon, lead to a whole new perspective on physical processes as storing and processing information. It also lead to paradoxes such as Maxwell's demon which seemed to suggest that work could be generated from heat only with the use of information, which would violate the second law of thermodynamics (for a review, see Refs. \cite{bennett_thermodynamics_1982, maruyama_colloquium:physics_2009}). The paradox was solved independently by Penrose and Bennett, in considering the entropy creation caused by erasing information \cite{penrose_foundations_1970, bennett_thermodynamics_1982}.

With the insight that erasing information generates entropy, Zurek found limits on the thermodynamic cost of deterministic computation using algorithmic complexity \cite{zurek_thermodynamic_1989}. A \emph{deterministic computation} generates a unique output given a particular input. Repeated computations yield identical results. A \emph{stochastic computation} on the other hand yields different outputs for identical inputs. It is a useful descriptor of natural processes which are often stochastic and have different final states given `identical' initial states (within a given finite resolution or, in quantum mechanics, even infinite resolution). Here, we derive the minimum energy cost of simulating a complex data set with a stochastic computation. We show that it is proportional to the difference between the \emph{statistical complexity}  \cite{crutchfield_inferring_1989} and the \emph{{\PI}} of the data \cite{bialek_predictability_2001}. We derive this difference as the amount of information erased during the computation. Finally, we illustrate the physics of information by ``implementing'' the stochastic computation with a Szilard-type engine.

It has been shown that the difference between these two measures arises from an asymmetry in the transport of information forward and backward in time \cite{crutchfield_times_2009}. In this paper we give a physical explanation of this asymmetry together with new mathematical proofs of the relevant information theory.

Consider the following model of stochastic computation: A given computational device is in some initial state and outputs length-$N$ strings of symbols $x^N$ according to probability distribution $\Pr(X^N)$. If the distribution $\Pr(X^N)$ is statistically indistinguishable from that of an observed data set of a physical system the symbol sequence $x^N$ is a simulation of that system. Where the probability distribution $\Pr(X^N)$ is a very uncompressed description, one step away from raw data of an experiment, the computational device simulating it is a very compact description, a summary of the regularities, a first step toward a `theory' explaining an experiment. We call the joint probability distribution of past and future observations $\Pr(X_{-N}^{N-1})$ a \emph{stochastic process}.
The provably unique minimal (in terms of information stored) and optimal (in terms of prediction) such computation-theoretic representation summarising the regularities of a stochastic process is a so-called {\eM} \cite{crutchfield_inferring_1989, shalizi_computational_2001}. It consists of an output alphabet $\Abet$,  a set of \emph{causal states} $\States$ and stochastic transitions between them. For every pair of states $s,s^\prime \in \States$ probabilities $\Pr(\nextS =\state^\prime | \lastS=\state, X_{i} = x)$ are defined  for going from state $s$ to state $s^\prime$ while outputting symbol $x \in \Abet$. The \emph{statistical complexity} of a process is defined as the Shannon entropy over the stationary distribution of its {\eM}'s causal-states  \footnote{For definitions of Shannon entropy,
conditional entropy, and mutual information, see for example \cite{cover_elements}}:
\begin{align}
\Cmu := H(\States)~.
\end{align}
$\Cmu$ is the
number of bits required to specify a particular causal state in the {\eM}. It is the number of bits that need to be stored to optimally predict future data points.

The \emph{\PI} of a data set is given by the mutual information between the two halves, e.g. the past data and the future data 
\cite{shalizi_computational_2001,
crutchfield_regularities_2003}:
\begin{align}
\label{eq.EE}
\EE = \lim_{N\rightarrow \infty} I[\past ; \future]~.
\end{align}
where $\past$ and $\future$ are strings of random variables representing observations of a stochastic process. Predictive information is also known under the name of excess entropy,  effective measure complexity, and stored information (see \cite{crutchfield_regularities_2003} and refs. therein). For the following thermodynamic development of stochastic computation we find the name predictive information most suitable.
The \PI, measured in bits, can be interpreted as the average number of bits a process stores at a given point in time and ``transmits'' to the  future. It is known that \begin{align}\label{eq.ECmu} \Cmu = \EE + H(S_{-1} | X_0^\infty)  \end{align} and hence that $ \EE \leq \Cmu~$ \cite[Theorem 5]{shalizi_computational_2001}.
Two important properties of {\eM}s of relevance here are that the next state given the last state and the current symbol is uniquely determined (Eq.~\ref{eq.determ}) and that the state after the observation of a long enough sequence of symbols is uniquely determined (Eq.~\ref{eq.sync})
\cite{shalizi_computational_2001}:
\begin{align}
\label{eq.determ}
H(S_{N} | S_{N-1}X_N) &= 0~~\text{(deterministic-stochastic)}\\
\label{eq.sync}
\lim_{N\rightarrow \infty}H(S_N | X_0^N) &= 0~~\text{(synchronising)}
\end{align}

Successful inference of {\eM}s ranges from dynamical systems \cite{crutchfield_inferring_1989}, spin systems \cite{crutchfield_statistical_1997}, and crystallographic data \cite{varn_discovering_2002} to molecular dynamics \cite{li_multiscale_2008}, atmospheric turbulence \cite{palmer_complexity_2000}, and self-organisation \cite{shalizi_quantifying_2004}.\\

Landauer defines an operation to be \emph{logically irreversible} if the output of the operation does not uniquely define the inputs \cite{landauer_irreversibility_1961}. In other words, logically irreversible operations erase information about the computational device's preceding logical state. Landauer's insight was that logical information erasure costs energy \cite{landauer_irreversibility_1961}. In the following we discuss how Landauer's principle and logical irreversibility apply to stochastic computation and, in particular, to the computation of a stochastic process. For a given {\eM} the current state and the next symbol determine the next state uniquely (Eq.~\ref{eq.determ}). The reverse, however, is not necessarily true. Given the current state and last output symbol, the previous state is not always uniquely determined. In this case the {\eM} is \emph{logically irreversible}. Following Landauer's definition of irreversibility, we define the information erasure per computational step of a given {\eM}  as the entropy of the previous state ($\lastS$) given the current state ($\nextS$) and last output symbol ($X_i$) :
\begin{equation}
\label{eq.ier}
\ier := H(\lastS | X_i {\State}_i),
\end{equation}
For later use, we also define $\ier^N := H(\State_{i-N} | X_{i-N}^i {\State}_i)$ for $i>N$ and $\LL := \lim_{N\to\infty} \ier^N$.
$\ier$ can be calculated from the {\eM} directly. It quantifies the average \emph{irreversibility} of a computational step of the {\eM}. We now show that strict equality between $\EE$ and $\Cmu$ holds if and only if the {\eM} is fully logically reversible, i.e. {\it iff} $\ier = 0$.

\begin{thm}
 The {\PI} $\EE$ of a stochastic process is equal to the statistical complexity $\Cmu$  if and only if the information
erasure of the corresponding {\eM} is zero:
\begin{equation}
\label{eq.cec}
\Cmu = \EE \Leftrightarrow \ier = 0~.
\end{equation}
\end{thm}

\begin{proof}

{\bf '$\Rightarrow$': }

From the Markov property of the states $S_i$ (i.e. $H(X_1 | S_{-1} X_0 S_0) = H(X_1 | X_0 S_0)$) it follows that $H(S_{-1} | X_0 X_1 S_0) = H(S_{-1} | X_0 S_0)$.
Using this recursively and the fact that further conditioning never increases the entropy we obtain the forward direction
\begin{align}
\notag
\Cmu = \EE &\Leftrightarrow H(S_{-1} | X_0^\infty) = 0 \\
\notag
&\Rightarrow H(S_{-1} | X_0^\infty S_0) = 0 \\
\label{eq.4}
&\Rightarrow H(S_{-1} | X_0 S_0) = 0
\end{align}

{\bf '$\Leftarrow$':}
Going in reverse we 
have $ H(S_{N-1} | X_N S_N) = 0 ~ \Rightarrow H(S_{N-1} | X_0^N S_N) = 0$ and in addition
\begin{align}
H(S_{N-1} | X_0^N S_N) = 
 H(S_{N-1} | X_0^N) - H(S_{N} | X_0^N)~.
\end{align}
The second term on the RHS goes to zero in the limit $N\to \infty$  (Eq.~\ref{eq.sync}). 
In the same way $H(S_{N-k} | X_0^N S_{N-k+1}) \to H(S_{N-k+1} | X_0^N)$, $ k=2,3,\dots,N+1$, as $N\to \infty$. Setting $k=N+1$, the claim follows.
\end{proof}

Note that this result automatically implies that any {\eM} with $ \ier = 0$ can be inverted in time, turning it into a {\it retrodicter} as introduced in \cite{crutchfield_times_2009} and thus providing an immediate and quick construction of a retrodicter for this case saving the `modeler' a second computationally costly inference procedure (for more on computational cost see e.g. \cite{dana_complexity_1978}).

For future reference we call $H(X_{i} | {\State}_{i}):=\hR $ the uncertainty of the last symbol given the current state. This reverse entropy rate is different from the reverse entropy rate referred to in dynamical systems theory (see e.g. \cite{gaspard_time-reversed_2004}). The complimentary quantity to $\hR$ is the well-known entropy rate of a stochastic process $\h = H(X_{i+1} | {\State}_{i})$. We can see that the amount of information which can be erased per computational step is upper bounded by the amount of information which is created per computational step,
\begin{align}
\ier \leq \h~,
\end{align}
by writing the joint entropy $H({\State}_{i-1} X_{i} {\State}_i)$ as two different sums:
\begin{align}
\notag
H({\State}_{i} | {\State}_{i-1} X_i) &+ H({\State}_{i-1} X_i) =H({\State}_{i-1} | X_{i} {\State}_i) + H(X_{i} {\State}_i) \\
\notag
\Leftrightarrow ~~&H(X_{i} | {\State}_{i-1}) + H({\State}_{i-1})  \\
\notag
&= H({\State}_{i-1} | X_{i} {\State}_i) + H(X_{i} | {\State}_i) + H({\State}_i)\\
\label{eq.HSXS}
\Leftrightarrow ~~& h^R = \ier + h ~,
\end{align}
where, in the second line, we have used Eq.~\ref{eq.determ}. 

Now consider a  {\eM} ${\mathcal M}$ with $\ier > 0$. To derive the difference between $\EE$ and $\Cmu$ we construct an {\eM} which outputs more than one symbol at a time as follows. For every pair of states $\state_i,\state_j \in \States$ of {\eM} ${\mathcal M}$  we construct the $N^{th}$ concatenation ${\mathcal M}^{\otimes N}$  with state transition probabilities 
$\Pr(\state_j | \state_i x^N) = \sum_{\state_1^{N-1}\in \State^{N-1}} \Pr(\state_{j} | \state_{N-1} x_{N}) $ $\left[ \prod_{k=2}^{N} \Pr(\state_{k} | \state_{k-1} x_{k}) \right] $ $\cdot \Pr(\state_1 | \state_i x_1)$
upon outputting $x^N$, where the sum runs over all state sequences $\state_1^{N-1}$ of length $N-1$.
Note, that ${\mathcal M}$ and ${\mathcal M}^{\otimes N}$ have the same set of states and the same probability distribution over output strings $\Pr(X^N)$ -- so they have the same $\EE$ and $\Cmu$. 

\begin{thm}
\label{thm.erase}
 The difference between the statistical complexity $\Cmu$ and the {\PI} $\EE$ of a stochastic process approaches the information erased by the corresponding concatenated {\eM} ${\mathcal M^{\otimes N}}$ as $N\to\infty$:
$$\lim_{N\to\infty} \ier^N = \Cmu - \EE~. $$
\end{thm}

\begin{proof}
By rewriting $H(\State_{-1} X_0^N \StateN)$ in two different ways we obtain
the information-theoretic equality: 
\begin{align}
\notag
H(\Sone | X_0^N) = &H(\Sone | X_0^N \State_N) \\
	&+ H(\State_N | X_0^N) - H(\State_N |\State_{-1}X_0^N)
\end{align}
The last term is zero due to determinism (Eq.~\ref{eq.determ}), the second term goes to zero for $N\rightarrow
\infty$. Hence, taking the limit we obtain:
\begin{align}
H(\State_{-1} | X_0^N) \rightarrow H(\State_{-1} | X_0^N \State_N) \text{ as } N\rightarrow\infty~.
\end{align}
The term on the right-hand side is exactly $\ier$ of ${\mathcal M^{\otimes
N}}$. Letting $N\to\infty$ the claim follows.
\end{proof}

Using Theorem~\ref{thm.erase} we can derive the minimum energy cost of simulating a stochastic process. Fig.~1 schematically illustrates an {\eM} contained in a box which on consecutive time steps outputs symbols  visible to an outside observer. The computational steps are as follows. In (a) the {\eM} in the box is in state $\State_{i-1}$. Going to (b) it generates symbol $X_i$ according to the probability distribution $\Pr(X_i | S_{i-1})$, leading to an increase in entropy inside the box by $h = H(X_i | \State_{i-1})$. Going to (c) the {\eM} moves from state $\State_{i-1}$ to state $\State_i$. Erasure of the previous-state information causes a decrease in entropy inside the box by $H(\State_{i-1} | \State_i X_i) = \ier $.  Finally, in (c) the symbol is ejected into the environment which decreases the entropy inside the box again, this time by $H(X_i | S_i) = h^R$. This closes one cycle of computation. The entropy contributions during one closed cycle must add up to zero and we obtain $\h - \ier - \hR = 0$ which is exactly Eq.~\ref{eq.HSXS}.

We now modify this stochastic computation by allowing for the generation of $N+1$ symbols at a time. The {\eM} inside the box (Fig.~1) 
is replaced by the concatenated {\eM} ${\cal M}^{\otimes (N+1)}$. In (a) this machine starts out in state $\State_{i-1}$, going to (b) it generates $N+1$ symbols according to $\Pr(X_{i}^{i+N} | S_{i-1})$. This causes an increase in entropy inside the box by $H(X_{i}^{i+N} | S_{i-1})$. Going to (c) the concatenated {\eM} moves to state $\State_{i+N}$ erasing the previous-state information which decreases the entropy inside the box by $H(\State_{i-1} | X_i^{i+N} S_{i+N})$. Ejecting the symbol sequence into the environment  the entropy of the box decreases by $H(X_i^{i+N} | S_{i+N})$. 
Setting w.l.o.g. $i=0$, we obtain for the entropy balance of one computational cycle:
\begin{align}
 H(X_0^{N} | \State_{-1}) - H(\State_{-1} | X_0^{N} \State_{N}) - H(X_0^N | \State_N)= 0~.
\end{align}
The LHS can be rewritten as 
\begin{align} H(\State_{-1}) - H(\State_{-1} | X_0^N \State_N) - I(\State_{-1};X_0^N)  
\end{align}

Letting $N\to \infty$ we obtain 
\begin{align}
\LL + \EE - \Cmu= 0~.
\end{align}

Hence, the entropy balance of one cycle of stochastic computation is in the limit of an infinite string of output given by Eq.~\ref{eq.ECmu} and Theorem~\ref{thm.erase}.
We now discuss the thermodynamics of such a stochastic computation.

\begin{figure}
\begin{center}
\label{fig.cycle}
\includegraphics[scale=.4]{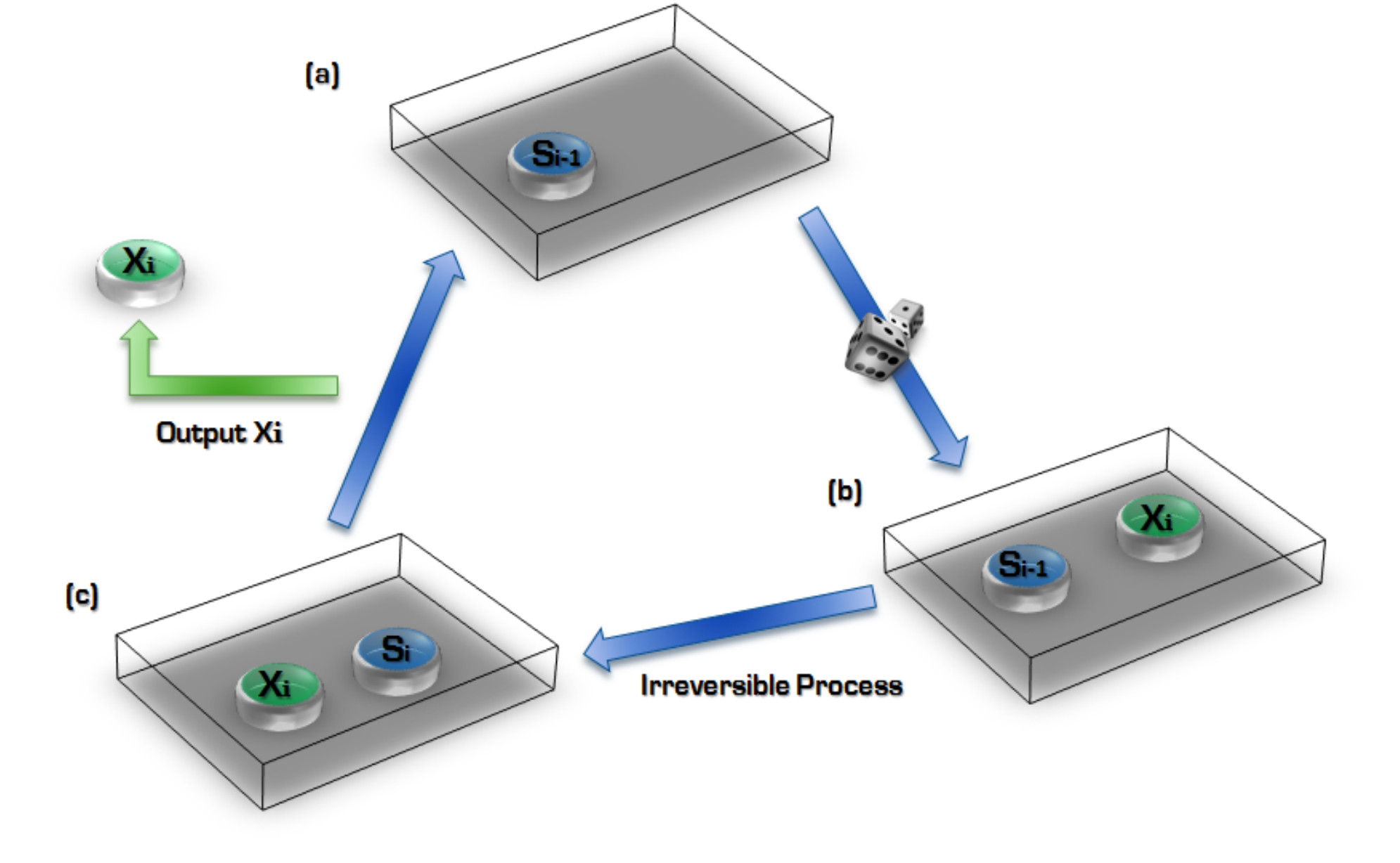}
\caption{One computational cycle for simulating a stochastic process with an {\eM}.}
\end{center}
\end{figure}

The Carnot efficiency of one cycle of an engine consisting of an ideal gas in a cylinder alternately connected to a hot and a cold reservoir at temperature $T_H$ and $T_C$, respectively, is, as is well known, given by the ratio of work output $W$ and absorbed heat $Q_H$:
\begin{align}
\label{eq.eta}
\eta = \frac{W}{Q_{H}} = 1-\frac{T_{C}}{T_{H}}~.
\end{align}
In 1929, Szilard invented a Gedankenexperiment of a single-particle engine to resolve Maxwell's demon paradox which seemed to defy the second law of thermodynamics \cite{szilard_entropieverminderung_1929}. The engine consists of a particle in a box, a measurement device locating the particle in either half of the box, and a memory to store the measurement result. Szilard considered the following procedure for extracting work from the particle's thermal motion. A user measures the particle's position and stores this one bit of information in the memory. Subsequently she 'compresses' the box to the half which contains the particle. This does not require any work. The thermal motion of the particle `decompresses' the box again and hence lets the user extract work from the box without any cost. This apparent paradox is resolved when one factors the additional energy required to for the user to reset her memory \cite{landauer_irreversibility_1961}. With the memory initially at temperature $T_H$ and  the reset done at temperature $T_C$ the efficiency of this engine is given by the Carnot efficiency (Eq.~\ref{eq.eta}) and the laws of thermodynamics were restored.

This Gedankenexperiment can be extended to a particle in one of $2^\EE$ possible partitions. Hence, storing the measurement result of the particle's  position requires $E$ bits of memory. Following the same argument as for the original Szilard engine we obtain Carnot efficiency $\eta = 1-T_{C}/T_{H}$. This modified engine leads us directly to a new interpretation of Theorem~\ref{thm.erase}.

In the simulation of a stochastic process, we attempt to generate information about its future based on observations of the past. This may be viewed as a Gedankenexperiment where we attempt to maximize our knowledge about a particle whose position is governed by a random variables $X_0,X_1,...$ which we can only indirectly measure by recording appropriate information from a correlated random variables $X_{-1},X_{-2},...$. These recorded bits can the be translated to information about $X_0,X_1,...$  by the use of an appropriate simulator. To extract the maximum possible amount of information, $\EE$, the theory of {\eM}s dictates that we must record at least record $\Cmu$ bits. The minimality and optimality of {\eM}s ensures that any fewer bits would render the simulation sub-optimal. This results in a stochastic computation that allows to extract

$Q_H^{sc} =  kT_H  \EE \ln 2$ units of extra work. Meanwhile, the $\Cmu$ bits stored about the past  are erased, which costs $Q_C^{sc} = kT_C \Cmu \ln 2$ units of energy. 
The efficiency is, just like before, the ratio between output work and absorbed heat:
\begin{align}
\frac{W^{sc}}{Q_H^{sc}} = 1  -\frac{Q_C^{sc}}{Q_H^{sc}} &= 1 - \frac{\Cmu T_C}{\EE T_H}~.
\end{align}
We define the information-theoretic efficiency for computing a stochastic process:
\begin{align}
\iota := \frac{\EE}{\Cmu}  = 1-\frac{\LL}{\Cmu}~.
\end{align}
Combining the thermodynamic and information theoretic efficiencies we obtain 
\begin{align}
\frac{W^{sc}}{Q_H^{sc}} = \eta (\iota) = 1 - {\iota^{-1}} \frac{T_H}{T_C}~.
\end{align}
For maximal information-theoretic efficiency we recover the thermodynamic efficiency from Eq.~\ref{eq.eta}. 

$\EE/\Cmu$ has been named the ``predictive efficiency of a process as the fraction of the information it contains which actually effects the future''  \cite{shalizi_what_2003}. Our results supply this concept with physicality and mathematical rigour.

We have derived the minimum energy cost of simulating a physical system as the difference between two information-theoretic complexity measures of the data. Of the two measures, the {\PI} measures the amount of information stored about a process's past transmitted to the future, the statistical complexity measures the amount of information required to compute this future. Any difference between the two is given by the amount of information erased during the simulation of the data and hence represents the minimum energy cost of physically running a simulation.

This result is complementary to the discussion of Crutchfield et al. who derive the difference between the two measures from the asymmetry of running the process in forward and reverse \cite{crutchfield_times_2009}. We add to this a physical interpretation of the cost of reversing a computation using thermodynamics. The lower bound to the energy cost of simulating a physical system was derived for optimal classical simulators. Recent results that quantum simulators require less information storage could indicate that quantum information leads to a reduced cost for stochastic computation \cite{gu_sharpening_2011}. 
Our results reveal an intricate relation between thermodynamics, information processing, and complexity. They motivate the use of information-theoretic tools for studying the physics of complex systems. 

\bibliography{lecs}

\end{document}